\documentclass[11pt]{article}

\usepackage{amsmath}
\usepackage{amssymb}
\usepackage{amsfonts}
\usepackage{amsthm}
\usepackage[margin=1in]{geometry}
\usepackage{hyperref}
\usepackage[width=12cm,format=plain,labelfont=sf,bf,small,textfont=sf,up,small,justification=justified,labelsep=period]{caption}
\usepackage{graphicx}
\usepackage{framed}
\usepackage{xcolor}
\usepackage{tabularx}

\hypersetup{
    colorlinks=true,       
    linkcolor=blue,          
    citecolor=blue,        
    filecolor=blue,      
    urlcolor=blue           
}

\makeatletter
\def\iddots{\mathinner{\mkern1mu\raise\p@
\vbox{\kern7\p@\hbox{.}}\mkern2mu
\raise4\p@\hbox{.}\mkern2mu\raise7\p@\hbox{.}\mkern1mu}}
\makeatother

\newcommand{\NN}{\mathbb{N}}
\newcommand{\RR}{\mathbb{R}}
\newcommand{\CC}{\mathbb{C}}

\DeclareMathOperator{\supp}{supp}

\DeclareMathOperator{\Tr}{Tr}
\DeclareMathOperator{\conv}{{\bf conv}}

\DeclareMathOperator{\linspan}{span}

\DeclareMathOperator{\cl}{{\bf cl}}

{\begin{small}\begin{list}{}%
         {\setlength{\leftmargin}{1cm}\setlength{\rightmargin}{1cm}}%
         \item[]%
}
{\end{list}\end{small}}

\renewcommand{\tilde}[1]{\widetilde{#1}}

\theoremstyle{plain}
\newtheorem{proposition}{Proposition}

\newtheorem{theorem}{Theorem}

\theoremstyle{definition}

\theoremstyle{remark}

\newtheorem{remark}{Remark}

\newcommand{\Sep}{\text{Sep}}
\newcommand{\Sepcone}{\mathcal{SEP}}
\newcommand{\PPT}{\text{PPT}}
\newcommand{\PPTcone}{\mathcal{PPT}}
\newcommand{\DPS}{\text{DPS}}

\newcommand{\C}{\mathcal{C}}
\renewcommand{\H}{\mathbf{H}}
\newcommand{\M}{\text{M}} 

\newcommand{\m}{\mathbf{m}} 

\newcommand{\psd}{\geq}

\renewcommand{\Re}{\text{Re}}
\renewcommand{\Im}{\text{Im}}

\newcommand{\T}{\mathsf{T}}

\newcommand{\W}{\mathbf{W}}

\newcommand{\pp}{\tilde{p}}

\newenvironment{sm}{\left[\begin{smallmatrix}}{\end{smallmatrix}\right]}

\title{The set of separable states has no finite semidefinite representation except in dimension $3\times 2$}
\author{Hamza Fawzi\thanks{Department of Applied Mathematics and Theoretical Physics, University of Cambridge. Email: \texttt{h.fawzi@damtp.cam.ac.uk}.}}
\date{May 4, 2019}

\begin{document}

\maketitle

\begin{abstract}
Given integers $n \geq m$, let $\Sep(n,m)$ be the set of separable states on the Hilbert space $\CC^n \otimes \CC^m$. It is well-known that for $(n,m)=(3,2)$ the set of separable states has a simple description using semidefinite programming: it is given by the set of states that have a positive partial transpose. In this paper we show that for larger values of $n$ and $m$ the set $\Sep(n,m)$ has no semidefinite programming description of finite size. As $\Sep(n,m)$ is a semialgebraic set this provides a new counterexample to the Helton-Nie conjecture, which was recently disproved by Scheiderer in a breakthrough result. Compared to Scheiderer's approach, our proof is elementary and relies only on basic results about semialgebraic sets and functions.
\end{abstract}


\section{Introduction}

Entanglement is a fundamental aspect of quantum mechanics. The set of \emph{separable states} (i.e., nonentangled states) on the Hilbert space $\CC^{n} \otimes \CC^{m}$ is defined as:
\[
\Sep(n,m) = \conv \left\{ x x^{\dagger} \otimes yy^{\dagger} : x \in \CC^n, |x| = 1, y \in \CC^m, |y| = 1\right\}.
\]
Here $x^{\dagger} = \bar{x}^{\T}$ indicates conjugate transpose, $|x|^2 = x^{\dagger} x = \sum_{i=1}^n |x_i|^2$ and $\conv$ denotes the convex hull. The set $\Sep(n,m)$ lives in the space $\H^{nm}$ of Hermitian matrices of size $nm \times nm$, and it is full-dimensional in the subspace of matrices of trace equal to one.

A fundamental computational task in quantum information is to decide membership in the convex set $\Sep(n,m)$. One of the first tests designed to check whether a state $\rho \in \H^{nm}$ is separable is the \emph{Peres-Horodecki criterion} \cite{peres1996separability,horodeckiPPT} (also known as the \emph{Positive Partial Transpose (PPT)} criterion). It is based on the observation that for any $\rho \in \Sep(n,m)$, $(I\otimes \T)(\rho)$ is positive semidefinite where $I$ is the identity map, and $\T$ the transpose map. Indeed one can easily verify that if $\rho = xx^{\dagger} \otimes yy^{\dagger}$ then $(I \otimes \T)(\rho) = xx^{\dagger} \otimes (yy^{\dagger})^{\T} = xx^{\dagger} \otimes \bar{y} \bar{y}^{\dagger} \psd 0$.  In other words we have the inclusion $\Sep(n,m) \subseteq \PPT(n,m)$  where
\begin{equation}
\label{eq:PPTnm}
\PPT(n,m) = \Bigl\{ \rho \in \H^{nm} : \rho \psd 0, \;\; (I \otimes \T)(\rho) \psd 0, \text{ and } \Tr[\rho]=1\Bigr\}.
\end{equation}
It is known, from earlier work of Woronowicz \cite{woronowicz1976positive}, that we have equality $\Sep(n,m) = \PPT(n,m)$ if, and only if $n+m\leq 5$. Thus, the smallest cases where $\Sep(n,m) \neq \PPT(n,m)$ are $(n,m) = (4,2)$ and $(n,m) = (3,3)$. 

\paragraph{Semidefinite programming} The description of the set $\PPT(n,m)$ in Equation \eqref{eq:PPTnm} allows us to decide membership, and optimize linear functions on $\PPT(n,m)$, via \emph{semidefinite programming}. Semidefinite programming is a fundamental tool in optimization that has played a crucial role in recent developments in quantum information theory. We say that a convex set $C$ has a \emph{semidefinite representation (also called a semidefinite lift) of size $r$} if it can be expressed as 
\begin{equation}
\label{eq:spectshadow}
 C = \pi(S)
\end{equation}
where $\pi:\RR^D\rightarrow \RR^d$ is a linear map and $S \subset \RR^D$ is a convex set defined using a linear matrix inequality
\begin{equation}
\label{eq:spect}
S = \{w \in \RR^D : M_0 + w_1 M_1 + \dots + w_D M_D \psd 0\}
\end{equation}
where $M_0,\ldots,M_D$ are Hermitian matrices of size $r\times r$. A set $S$ of the form \eqref{eq:spect} is known as a \emph{spectrahedron}. In this paper we are most interested in when a semidefinite representation of finite size exists, and call this simply a semidefinite representation throughout. If a convex set $C$ admits a semidefinite representation, then optimizing a linear function on $C$ can be cast as a semidefinite program. Equation \eqref{eq:PPTnm} gives a semidefinite representation of $\PPT(n,m)$.

\paragraph{Horodecki's criterion} The set of separable states has the following well-known description due to the Horodeckis \cite{horodeckiPPT}: 
\begin{equation}
\label{eq:Sepnmhorodecki}
\Sep(n,m) = 
\left\{ \rho \in \H^{nm} : \Tr[\rho] = 1 \text{ and } (I \otimes \Phi)(\rho) \psd 0 \;\; \forall \Phi: \M_m \rightarrow \M_n \text{ positive}\right\}.
\end{equation}
Here $\M_k = \CC^{k\times k}$ and a $\CC$-linear map $\Phi : \M_m\rightarrow \M_n$ is \emph{positive} if it is Hermitian preserving and if $\Phi(X) \psd 0$ for all $X \psd 0$. When $n=m$, the relaxation $\PPT(n,m)$ corresponds to having only the identity and transpose maps in Equation \eqref{eq:Sepnmhorodecki}, which are both positive. A recent result of Skowronek \cite{skowronek} shows that when $n=m=3$, there is no finite family of positive maps $\Phi_1,\ldots,\Phi_k:\M_3 \rightarrow \M_3$ such that $\Sep(3,3) = \left\{ \rho \in \H^{9} : (I \otimes \Phi_i) (\rho) \psd 0 \; \forall i=1,\ldots,k\right\}$. Note that the right-hand side of the previous equation is a specific semidefinite representable set. Thus Skowronek's result rules out certain specific semidefinite representations for $\Sep(3,3)$.\footnote{The result of Skrownek is in fact more general than this, and rules any formulation of the form $\left\{ \rho : (I \otimes \Phi_i)((I\otimes B)\rho(I\otimes B^{\dagger})) \geq 0 \; \forall i=1,\ldots,k, \; \forall B \in \M_3 \right\}$ where $\Phi_1,\ldots,\Phi_k$ is a finite set of positive maps.}

\paragraph{DPS hierarchy} In \cite{doherty2004complete}, Doherty, Parrilo and Spedalieri proposed a complete hierarchy of approximations to the set of separable states based on semidefinite programming. The first level of the hierarchy coincides with the PPT test, and subsequent levels form tighter and tighter convex relaxations of the set of separable states. If we denote the convex relaxation at level $k$ by $\DPS_k(n,m)$ we have (dropping the $(n,m)$):
\[
\Sep \subseteq \dots \subseteq \DPS_k \subseteq \DPS_{k-1} \subseteq \dots \subseteq \DPS_1 = \PPT.
\]
The key property of the DPS hierarchy is that each set $\DPS_k$ has a semidefinite representation of size $\min(n,m)^{O(k)}$.  The hierarchy is known to be complete, meaning that if $\rho \notin \Sep$, then there exists a finite $k$ such that $\rho \notin \DPS_k$. The integer $k$ however depends on the state $\rho$ and it is known that, unless $n+m\leq 5$, there is no finite $k$ such that $\DPS_k(n,m) = \Sep(n,m)$ \cite[Section VIII.B]{doherty2004complete}.

\paragraph{Contributions} 
The main result of this paper is
\begin{theorem}
\label{thm:main}
If $\Sep(n,m) \neq \PPT(n,m)$ then $\Sep(n,m)$ has no (finite) semidefinite representation. In other words, $\Sep(n,m)$ has no semidefinite representation when $n+m > 5$.
\end{theorem}
\begin{remark}
Some remarks concerning Theorem \ref{thm:main}:
\begin{itemize}
\item Note that Theorem \ref{thm:main} contains as a special case the fact that whenever $\Sep(n,m) \neq \PPT(n,m)$, then there is no representation of $\Sep(n,m)$ as $\{\rho \in \H^{nm} : \Tr[\rho] = 1 \text{ and } (I \otimes \Phi_i)(\rho) \psd 0, \; \forall i=1,\ldots,k\}$ where $\Phi_1,\ldots,\Phi_k:\M_m\rightarrow \M_n$ is a finite family of positive maps. This is because the latter set is semidefinite representable (in fact it is a spectrahedron). Let us mention that if one is only interested in \emph{approximating} the set $\Sep(n,m)$,  Aubrun and Szarek \cite{aubrun2017dvoretzky} gave a lower bound on the number $k$ of positive maps needed.
\item Our result also includes as a special case the fact that there is no finite $k$ such that $\DPS_k(n,m) = \Sep(n,m)$, when $n+m > 5$. We note that our result is a strict generalization of this fact. Indeed, the failure of the DPS hierarchy to converge in a finite number of levels does not preclude by itself the existence of another semidefinite program that represents $\Sep(n,m)$ exactly. There are well-known examples of convex sets where the sum-of-squares hierarchy (of which DPS can be seen as a particular instance) is never exact and yet a finite semidefinite representation does exist, see e.g., \cite[Example 3.7]{netzer2010exposed}.
\item Observe that if $\Sep(n,m)$ has no semidefinite representation, then the same is true for $\Sep(N,m)$ for $N \geq n$. This is because $\Sep(n,m)$ can be realized as a linear section of $\Sep(N,m)$ as follows:
\[
\Sep(n,m) \simeq \left\{ \rho \in \Sep(N,m) : (\Tr_2 \rho)_{ii} = 0 \; \forall i=n+1,\ldots,N \right\}
\]
where $\Tr_2 \rho$ is the result of tracing out the second subsystem from $\rho$. Indeed, setting $(\Tr_2 \rho)_{ii} = 0$ implies that in any representation of $\rho$ as $\rho = \sum_{k} p_k x_k x_k^{\dagger} \otimes y_k y_k^{\dagger}$, the vectors $x_k$ must satisfy $(x_{k})_{i} = 0$ for all $i=n+1,\ldots,N$, i.e., that $x_k \in \CC^n \times \{0\}^{N-n} \simeq \CC^n$. To prove our theorem it thus suffices to prove that $\Sep(3,3)$ and $\Sep(4,2)$ have no semidefinite representations.
\end{itemize}
\end{remark}

\paragraph{Helton-Nie conjecture} The question of finding semidefinite representations for convex sets has attracted a lot of attention in the optimization community \cite{nemirovskiICM,gouveia2011lifts}. Helton and Nie \cite{helton2009sufficient} gave sufficient conditions for a set to have a semidefinite representation, and conjectured that any convex \emph{semialgebraic} set has a semidefinite representation. (A set is semialgebraic set if it can be described using a finite boolean combinations of polynomial equations and inequalities. One can verify that $\Sep(n,m)$ is a semialgebraic set, see Section \ref{sec:bgsa}.) In his breakthrough paper, Scheiderer  \cite{scheidererSDR} disproved this conjecture and exhibited convex semialgebraic sets that have no semidefinite representations.

Our proof of Theorem \ref{thm:main} is inspired from the arguments of Scheiderer. Compared to the paper of Scheiderer the present paper has two main contributions. First, the proof we give simplifies the arguments of Scheiderer and does not rely on any specialized results from algebraic geometry. We only use basic results from analysis (Taylor expansions), and some standard facts about semialgebraic sets and functions which are elementary to state. The proof should thus be accessible to readers in quantum information and optimization. The second contribution is the application of the method of proof for $\Sep(n,m)$ which is defined in terms of complex numbers. This turns out to cause certain difficulties as certain standard facts about real polynomials are not true about Hermitian polynomials, particularly on the relation between homogeneous polynomials and their dehomogenizations (see Section \ref{sec:proofmain} and Appendix \ref{sec:hakye} for more details). 

\paragraph{Main technical result} 
Our main technical result, Theorem \ref{thm:mainhermitian} below and of which Theorem \ref{thm:main} is a corollary, gives a general way to construct a convex set with no semidefinite representation from a nonnegative Hermitian polynomial that is not a sum of squares.
 We recall that a \emph{Hermitian polynomial} $p(z)$ is a polynomial with complex coefficients in the indeterminates $(z,\bar{z})=(z_1,\ldots,z_n,\bar{z}_1,\ldots,\bar{z}_n)$ such that $p(z) \in \RR$ for all $z \in \CC^n$. A Hermitian polynomial is a sum of squares if it can be written as a sum of squares of Hermitian polynomials. (More details about Hermitian polynomials are given in Section \ref{sec:prelim}.) For the statement of the theorem, we use the monomial notation $z^{u} = \prod_{i=1}^n z_i^{u_i}$ for $u \in \NN^n$.

\begin{theorem}[General theorem]
\label{thm:mainhermitian}
Let $p(z) = \sum_{(u,v) \in A} p_{uv} z^u \bar{z}^v$ be a Hermitian polynomial supported on $A \subset \NN^n\times \NN^n$, and assume that $p$ is nonnegative on $\CC^n$ but not a sum of squares. Assume furthermore that $A$ is \emph{downward closed}, i.e., if $(u,v) \in A$ then all $(u',v') \in \NN^n \times \NN^n$ with $0 \leq u' \leq u$ and $0\leq v' \leq v$ are in $A$. Define the monomial map $\m_A:\CC^n \rightarrow \CC^{|A|}$, $z\mapsto \left[ z^{u} \bar{z}^{v}\right]_{(u,v) \in A}$ for $z \in \CC^n$. Then the convex set
\begin{equation}
\label{eq:CA}
\C_A = \cl \conv\left\{ \m_A(z) : z \in \CC^n \right\}
\end{equation}
is not semidefinite representable, where $\cl$ denotes topological closure.
\end{theorem}
The set of separable states is of the form \eqref{eq:CA} for well-chosen set $A$. Indeed, dropping the normalization condition and letting $\Sepcone(n,m)$ be the convex \emph{cone} of separable states, we have:
\[
\begin{aligned}
\Sepcone(n,m) &= \conv \left\{ \left[x_i \bar{x}_j y_k \bar{y}_l\right]_{\substack{1 \leq i,j \leq n\\ 1 \leq k,l \leq m}} : (x,y) \in \CC^n \times \CC^m \right\}\\
 &= \conv \left\{ \left[x^{\alpha} \bar{x}^{\beta} y^{\gamma} \bar{y}^{\delta}\right]_{|\alpha|=|\beta|=|\gamma|=|\delta|=1} : (x,y) \in \CC^n \times \CC^m \right\}
 \end{aligned}
\]
where for $\zeta \in \NN^k$ we let $|\zeta| = \sum_{i=1}^k \zeta_i$. This shows that $\Sepcone(n,m) = \C_A$ where
\begin{equation}
\label{eq:Asep}
 A = \left\{(\alpha,\beta,\gamma,\delta) \in (\NN^n \times \NN^n) \times (\NN^m \times \NN^m) : |\alpha|=|\beta|=|\gamma|=|\delta|=1\right\}.
 \end{equation}
The attentive reader will notice that this set $A$ is \emph{not} downward closed, and so does not satisfy the condition of Theorem \ref{thm:mainhermitian}. As a matter of fact, to prove Theorem \ref{thm:main} we apply Theorem \ref{thm:mainhermitian} with a \emph{dehomogenization} of $A$ which satisfies the downward closed condition, and then homogenize back to get the desired convex cone. The details are explained in Section \ref{sec:proofmain}.

\paragraph{Overview of proof} We briefly sketch the main ideas for the proof of Theorem \ref{thm:mainhermitian}.
\begin{itemize}
\item We first show that if the set $\C_A$ has a semidefinite representation, then there exists a finite number of functions $f_1,\ldots,f_r:\RR^{2n} \simeq \CC^n \rightarrow \RR$ such that any nonnegative Hermitian polynomial supported on $A \cup \{(\mathbf{0},\mathbf{0})\}$ can be written as a sum of squares from $\linspan_{\RR}(f_1,\ldots,f_r)$.  This characterization of semidefinite representations via sums of squares is not new: it follows from the factorization theorem of Gouveia, Parrilo and Thomas \cite{gouveia2011lifts} and its sum-of-squares interpretation see e.g., \cite{fawziphdthesis}. We note that a similar characterization is also used in Scheiderer's paper, see \cite[Theorem 3.4]{scheidererSDR}. 
\item One of the main observations needed to prove Theorem \ref{thm:mainhermitian} is to note that the functions $f_1,\ldots,f_r$ can be chosen to be \emph{semialgebraic}. (We recall the precise definition of semialgebraic functions in Section \ref{sec:prelim}.) One key property of such functions that turns out to be particularly important is that they are smooth almost everywhere. Combining this property with a simple observation regarding smooth sum of squares decompositions of homogeneous polynomials allows us to prove Theorem \ref{thm:mainhermitian} already in the special case where $p$ is a homogeneous polynomial. This allows us to prove that $\Sep(n,m)$ is not semidefinite representable when $(n,m) = (5,3)$ or $(4,4)$. The complete proof of Theorem \ref{thm:mainhermitian} which allows us to cover the cases $(n,m)=(4,2)$ and $(3,3)$ for separable states, requires an additional technical argument using Puiseux expansions for univariate continuous semialgebraic functions.
\end{itemize}

\paragraph{Real version of Theorem \ref{thm:mainhermitian}} We note that one can state an analogue of Theorem \ref{thm:mainhermitian} dealing with real polynomials instead of Hermitian polynomials. The proof is similar, and we state it below just for convenience and for future reference.

\begin{theorem}[Main theorem for real polynomials]
\label{thm:mainreal}
Let $p(x) = \sum_{u} p_u x^u \in \RR[x]$ where $A \subset \NN^n$ finite, be a real polynomial that is nonnegative on $\RR^n$ but not a sum of squares. Assume furthermore that $A$ is \emph{downward closed}, i.e., if $u \in A$ then all $u' \in \NN^n$ with $0 \leq u' \leq u$ are in $A$. Define the monomial map $\m_A(x) = \left[ x^{u} \right]_{u \in A}$ for $x \in \RR^n$. Then the convex set
\[
\cl \conv\left\{ \m_A(x) : x \in \RR^n \right\}
\]
is not semidefinite representable.
\end{theorem}

The theorem above can be used to recover the result of Scheiderer \cite[Corollary 4.25]{scheidererSDR}, that the cone $P_{n,2d}$ of nonnegative (real) forms in $n$ variables of degree $2d$ is not semidefinite representable when it is distinct from $\Sigma_{n,2d}$, the cone of sums of squares. Indeed, it suffices to take $p$ in Theorem \ref{thm:mainreal} to be a dehomogenization of a nonnegative form that is not a sum of squares, and to use the well-known fact that a convex set has a semidefinite representation if and only if its dual has one.

\paragraph{Organization} The paper is organized as follows. In Section \ref{sec:prelim} we set some of the notations and present some background material on Hermitian polynomials, sums of squares, and semialgebraic sets and functions that are useful for the proof of the main theorem. In Section \ref{sec:sdplifts} we review the connection between the existence of semidefinite programming representations, and sums of squares. The proof of Theorem \ref{thm:mainhermitian} is in Section \ref{sec:proofmainhermitian} and the proof of Theorem \ref{thm:main} in Section \ref{sec:proofmain}.

\section{Preliminaries}
\label{sec:prelim}

We recall in this section some results on Hermitian polynomials, the duality $\Sep$/nonnegative polynomials and $\PPT$/sums of squares and semialgebraic sets and functions.

\subsection{Hermitian polynomials}

For $z \in \CC^n$, we denote the elementwise complex conjugate of $z$ by $\bar{z} = (\bar{z}_1,\ldots,\bar{z}_n)$. If $u \in \NN^n$ we define the monomial $z^u = z_1^{u_1} \dots z_n^{u_n}$. A Hermitian polynomial $p(z)$ is a polynomial in $z$ and $\bar{z}$ of the form
\begin{equation}
\label{eq:hermpA}
p(z) = \sum_{(u,v) \in A} p_{uv} z^u \bar{z}^v \qquad (A \subset \NN^n \times \NN^n)
\end{equation}
such that $p(z) \in \RR$ for all $z \in \CC^n$. This is equivalent to saying that  $p_{uv} = \overline{p_{vu}}$ for all $u,v$. The \emph{support} of $p$ is $\supp(p) = \{ (u,v) : p_{uv} \neq 0\} \subset \NN^n \times \NN^n$. 
The Hermitian polynomial $p$ is \emph{nonnegative} if $p(z) \geq 0$ for all $z \in \CC^n$. Further, we say that $p$ is a \emph{sum of squares} if we can write
\begin{equation}
\label{eq:sosdef}
p = \sum_{k} q_k^2
\end{equation}
for Hermitian polynomials $q_k$.
If $p(z)$ is a Hermitian polynomial we will often consider the \emph{real} polynomial $P(a,b) = p(a+ib)$ in $\RR[a_1,\ldots,a_n,b_1,\ldots,b_n]$. One can check that $p$ is a sum-of-squares if and only if $P$ is a sum-of-squares of real polynomials.

\begin{remark}[Sums of squares for Hermitian polynomials]
Another common definition of a Hermitian polynomial $p(z)$ being a sum-of-squares is that $p$ can be written as $p(z) = \sum_{k} |g_k(z)|^2$ where $g_k$ are (holomorphic) polynomials in $z$ \emph{only} (and not in $\bar{z}$). Clearly if $p$ has such a representation then it is a sum-of-squares in the sense \eqref{eq:sosdef} since then $p = \sum_{k} \Re[g_k]^2 + \Im[g_k]^2$ and $\Re[g_k]$ and $\Im[g_k]$ are both Hermitian polynomials. The converse however is not true. It is possible that a polynomial $p$ has a representation \eqref{eq:sosdef} and cannot be written as a sum of modulus squares of holomorphic polynomial mappings. See e.g., \cite{putinarangelo} for more on this distinction. In this paper we only work with the definition \eqref{eq:sosdef} of sums of squares.
\end{remark}

\subsection{$\Sep$, $\PPT$, nonnegative polynomials, and sums of squares}

For convenience, we will work in this paper with the \emph{cone} of separable states, where we drop the normalization condition:
\[
\Sepcone(n,m) = \conv \left\{ x x^{\dagger} \otimes yy^{\dagger} : x \in \CC^n, y \in \CC^m \right\}.
\]
One can verify that $\Sep(n,m)$ is the compact slice $\Sep(n,m)=\Sepcone(n,m) \cap \{\rho : \Tr \rho = 1\}$.\footnote{Indeed if $\rho = \sum_{i} p_k x_k x_k^{\dagger} \otimes y_k y_k^{\dagger}$ with $\Tr \rho = 1$ and $p_k \geq 0$, then by redefining $p_k \leftarrow p_k |x_k|^2 |y_k|^2$ we can assume without loss of generality that $|x_k|=|y_k|=1$. Taking the trace on both sides of $\rho = \sum_{i} p_k x_k x_k^{\dagger} \otimes y_k y_k^{\dagger}$ tells us that $1 = \sum_k p_k$ since $\Tr \rho = 1$, i.e., $\rho \in \Sep(n,m)$.} Let also $\PPTcone$ be the cone of states that have positive partial transpose, i.e.,
\[
\PPTcone(n,m) = \left\{ \rho \in \H^{nm} : \rho \psd 0 \text{ and } (I \otimes \T)(\rho) \psd 0 \right\}
\]
so that $\PPT(n,m) = \PPTcone(n,m) \cap \left\{ \rho : \Tr \rho = 1\right\}$.

\paragraph{Dual of $\Sep$} For any integer $k$, let $\M_k = \CC^{k\times k}$. A $\CC$-linear map $\Phi : \M_n \rightarrow \M_m$ that is Hermitian preserving is \emph{positive} if $\Phi(\rho) \psd 0$ whenever $\rho \psd 0$. Equivalently, $\Phi$ is positive if the degree-four Hermitian polynomial $p(x,y) = y^{\dagger} \Phi\left( xx^{\dagger} \right) y$ is nonnegative on $\CC^{n+m} \simeq \CC^n \times \CC^m$. It is well-known that the dual of $\Sepcone(n,m)$ can be identified, via the Choi isomorphism, with the cone of positive maps $\M_n \rightarrow \M_m$ (see e.g., \cite[Table 2.2]{ABMB}).   Equivalently, the dual of $\Sepcone(n,m)$ can be identified with nonnegative degree-four Hermitian polynomials of the form
\begin{equation}
\label{eq:pdual}
p(x,y) = \sum_{\substack{1\leq i,j \leq n\\ 1\leq k,l \leq m}} p_{ijkl} x_i \bar{x}_j y_k \bar{y}_l \qquad (x \in \CC^n, y \in \CC^m)
\end{equation}
where $p_{ijkl} = \Phi(E_{ij})_{lk}$. Polynomials of the form \eqref{eq:pdual} have a \emph{biquadratic} structure: they are quadratic independently in each block of variables $x$ and $y$. The duality between $\Sepcone$ and nonnegative polynomials of the form \eqref{eq:pdual} is in fact immediate from the definition of $\Sepcone$.

\paragraph{Dual of $\PPT$} Using the identification above, it turns out that the dual of $\PPTcone(n,m)$ corresponds to polynomials $p(x,y)$ that are \emph{sums of squares}. Indeed, it is well-known (see again \cite[Table 2.2]{ABMB}) that the dual cone of $\PPTcone(n,m)$ can be identified, via the Choi isomorphism, with the cone of maps $\Phi : \M_n \rightarrow \M_m$ that are \emph{decomposable}, i.e., that can be written $\Phi = S_1 + S_2 \circ \T$ where $S_1$ and $S_2$ are two \emph{completely positive maps}, and $\T$ is the transpose map. Recall that a map $S:\M_n\rightarrow \M_m$ is completely positive if there exist matrices $V_t$ such that $S(X) = \sum_{t} V_t^{*} X V_t$. One can verify that a map $\Phi$ is decomposable if and only if, the associated Hermitian polynomial \eqref{eq:pdual} is a sum of squares. We did not find any reference for this equivalence, so we include a proof here. (The proofs we found in the literature are only for the direction $\Rightarrow$ in Proposition \ref{prop:decompsos}. The proof of Proposition \ref{prop:decompsos} is a special case of a more general result in \cite{fangfawzi}, joint with Kun Fang, where it is shown that the dual of $\DPS_k$ can be identified with a sum-of-squares condition of degree $k$.)

\begin{proposition}
\label{prop:decompsos}
A map $\Phi:\M_n \rightarrow \M_m$ is decomposable if, and only if, the Hermitian polynomial $p(x,y) = y^{\dagger} \Phi\left( xx^{\dagger} \right) y$ is a sum of squares.
\end{proposition}
\begin{proof}

If $S(\rho) = \sum_{t} V_t \rho V_t^{\dagger}$ is a completely positive map then $y^{\dagger} S(xx^{\dagger}) y = \sum_{t} y^{\dagger} V_t xx^{\dagger} V_t^{\dagger} y = \sum_{t} |\bar{y}^{\T} V_t x|^2$ is a sum-of-squares. Also for the transpose map $\T$, we have $y^{\dagger}(S \circ \T) ( xx^{\dagger} ) y = y^{\dagger} S(\bar{x} \bar{x}^{\dagger}) y = \sum_{t} |y^{\T} \bar{V_t} x|^2$ is also a sum-of-squares. It follows that if $\Phi$ is decomposable then $p(x,y) = y^{\dagger} \Phi\left( xx^{\dagger} \right) y$ is a sum-of-squares.

We now prove the converse. Assume $p(x,y) = y^{\dagger} \Phi(xx^{\dagger}) y$ is a sum-of-squares, i.e., $p(x,y) = \sum_{t} q_t(x,y)^2$ for some Hermitian polynomials $q_t$. We need to show that $\Phi$ is decomposable. Since the coefficient of the monomial $x_i^2 \bar{x_i}^2$ in $p$ is 0, we see that $q_t$ cannot have monomials $x_i^2, \bar{x_i}^2$ or $x_i \bar{x_i}$. To be sure, let $\alpha_t, \bar{\alpha_t}, \beta_t$ be the coefficients in $q_t$ of these monomials (note that $\beta_t \in \RR$ since $x_i \bar{x_i}$ is real). The coefficient of $x_i^2 \bar{x_i}^2$ in $\sum_t q_t^2$ is $\sum_{t} 2|\alpha_t|^2 + \beta_t^2 = 0$ which implies that $\alpha_t = \beta_t = 0$ for all $t$.  Similarly, by looking at the coefficient of $x_i \bar{x_i} x_j \bar{x_j}$ in $p$, we see that the $q_t$ cannot have monomials of the form $x_i x_j, \bar{x_i} \bar{x_j}, x_i \bar{x_j}$ or $ \bar{x_i} x_j$. The same of course is true for the $y$'s. Thus this means that $q_t$ must have the form
\[
q_t(x,y) = \underbrace{x^{\T} M_t y}_{g_t} + \underbrace{\bar{x}^{\T} \bar{M_t} \bar{y}}_{\bar{g_t}} + \underbrace{x^{\T} N_t \bar{y}}_{h_t} + \underbrace{\bar{x}^{\T} N_t y}_{\bar{h_t}}
\]
where $M \in \CC^{n\times n}$ and $N \in \CC^{m\times m}$. Squaring $q_t$ we get
\[
q_t^2 = g_t^2 + 2|g_t|^2 + 2g_t h_t + 2g_t \bar{h_t} + \bar{g_t}^2 + 2\bar{g_t} h_t + 2\bar{g_t} \bar{h_t} + h_t^2 + 2 |h_t|^2 + \bar{h_t}^2.
\]
When summing $\sum_{k} q_t^2$ we see that the only terms that can produce monomials of the form $x_i \bar{x_j} y_k \bar{y_{l}}$ (the monomials that appear in $p$) are the terms $2|g_t|^2$ and $2|h_t|^2$. The sum of all the other terms must thus be equal to 0. At the end we get (including the constant 2 in $M_t$ and $N_t$):
\[
p = \sum_{k} |x^{\T} M_t y|^2 + |x^{\T} N_t \bar{y}|^2.
\]
From here it easily follows that $\Phi = S_1 + S_2 \circ \T$ where $S_1(\rho) = \sum_{t} N_t \rho N_t^{\dagger}$ and $S_2(\rho) = \sum_{t} \bar{M_t} \rho \bar{M_t}^{\dagger}$.
\end{proof}

The following diagram summarizes the discussion above.

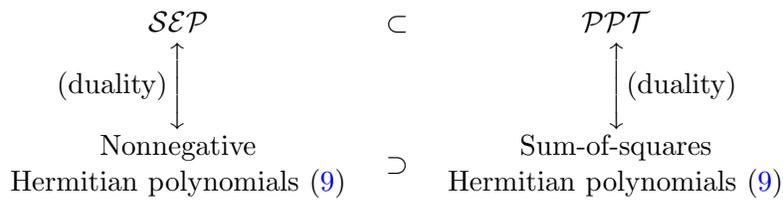
\[
\begin{array}{ccc}
\Sepcone & \subset & \PPTcone\\
\hspace*{-1.5cm}\text{(duality)}  \Bigg\updownarrow & & \hspace*{1.5cm} \Bigg\updownarrow \text{(duality)}\\
\begin{array}{c}\text{Nonnegative}\\\text{Hermitian polynomials \eqref{eq:pdual}} \end{array} & \supset &
\begin{array}{c}\text{Sum-of-squares}\\\text{Hermitian polynomials \eqref{eq:pdual}} \end{array}
\end{array}
\]

\subsection{Semialgebraic sets and functions}
\label{sec:bgsa}

\paragraph{Semialgebraic sets} A semialgebraic subset of $\RR^n$ is a subset that can be defined by a finite boolean combination of polynomial equations ($P = 0$) and inequalities ($P > 0$) where $P \in \RR[x_1,\ldots,x_n]$. For example a set of the form $\{w \in \RR^D : M_0 + w_1 M_1 + \dots + w_D M_D \psd 0\}$ is semialgebraic since the condition that a matrix is positive semidefinite can be expressed by a finite number of polynomial inequalities. The set of separable states can be shown to be semialgebraic. One can prove this using the celebrated and powerful result of Tarski stating that the projection of a semialgebraic set is semialgebraic. A consequence of Tarski's theorem is that the \emph{convex hull} of a semialgebraic set $S \subset \RR^n$ is semialgebraic. Indeed this is because we can write $\conv(S)$ as the projection on the $x$ component of the following semialgebraic set
\[
\begin{aligned}
\Bigl\{ & (x,\lambda,s_1,\ldots,s_{n+1}) \in \RR^n \times \RR^{n+1} \times (\RR^n)^{n+1} : \\
& \qquad \lambda_1,\ldots,\lambda_{n+1} \geq 0, s_1,\ldots,s_{n+1} \in S, \; x = \sum_{i=1}^{n+1} \lambda_i s_i \text{ and } \sum_{i=1}^{n+1} \lambda_i = 1 \Bigr\}.
\end{aligned}
\]
(Note that, by Carath\'eodory theorem any element in $\conv(S)$ is a convex combination of at most $n+1$ points in $S$.) To see why the set $\Sep(n,m)$ is a semialgebraic subset of $\H^{nm} \simeq \RR^{2 (nm)^2}$ first note that the following set
\begin{equation}
\label{eq:pureproducth}
\left\{ (\rho , x , y) \in \H^{nm} \times \CC^n \times \CC^m \text{ s.t. } \rho = xx^{\dagger} \otimes yy^{\dagger} \text{ and } |x|^2 = |y|^2 = 1 \right\}
\end{equation}
is a semalgebraic subset of $\H^{nm} \times \CC^n \times \CC^m \simeq \RR^{2(nm)^2} \times \RR^{2n} \times \RR^{2m}$ since the equations can all be written as real polynomial equations in the real and imaginary components. By Tarski's theorem it follows that the projection of \eqref{eq:pureproducth} on the $\H^{nm}$ component, which is precisely the set of pure product states, is semialgebraic. Thus $\Sep$ is semialgebraic as the convex hull of a semialgebraic set. 

\paragraph{Semialgebraic functions} A function $f:\RR^n\rightarrow \RR^m$ is called \emph{semialgebraic} if its graph $\{(x,f(x)) : x \in \RR^n\} \subseteq \RR^n\times \RR^m$ is a semialgebraic set. Even though semialgebraic functions form a very broad class of functions, they are tame and possess nice regularity properties. Examples of semialgebraic functions are polynomials, rational functions, or power functions (with rational exponent). Functions that are not semialgebraic are e.g., $\exp(x)$, or the indicator function of the rationals in $\RR$. We state two basic results about semialgebraic functions that will be crucial for us.

\begin{theorem}[Almost everywhere smoothness of semialgebraic functions]
\label{thm:sasmooth}
Let $f:\RR^n\rightarrow \RR$ be a semialgebraic function. Then $f$ is smooth ($C^{\infty}$) everywhere except possibly on the zero set of a polynomial $P \in \RR[x_1,\ldots,x_n] \setminus \{0\}$.
\end{theorem}
\begin{proof}
See e.g., \cite[Theorem 1.7]{son2016genericity}.
\end{proof}
Clearly Theorem \ref{thm:sasmooth} is not true for general functions, cf. the indicator function of the rationals in $\RR$.
The second result that we will need concerns semialgebraic functions in one variable.

\begin{theorem}[Puiseux expansion for one-dimensional semialgebraic functions]
\label{thm:puiseux}
Assume $f:(0,\eta) \rightarrow \RR$ where $\eta > 0$, is a semialgebraic continuous function that is bounded. Then $f$ can be extended by continuity to the interval $[0,\eta)$. Additionally, there exists an integer $m$ such that the map $t\mapsto f(t^m)$ is $C^{\infty}$ on $[0,\epsilon)$ for some $0 < \epsilon < \eta$.
\end{theorem}
\begin{proof}
For the first part, see \cite[Proposition 3.18]{ARAGbook}. For the second part, see \cite[page 10]{coste-RAS-notes}.
\end{proof}
We note that the theorem above is not true for arbitrary functions. For example the function $x\mapsto \sin(1/x)$ is bounded and continuous on any interval $(0,\eta)$ but cannot be extended by continuity at $0$. We finally record the following result which will also be needed for our proof. It simply says that any linear map, restricted to a semialgebraic set always admits a semialgebraic inverse.
\begin{theorem}
\label{thm:sachoice}
Let $S \subset \RR^N$ be a semialgebraic set and let $\pi:\RR^N \rightarrow \RR^n$ be a linear map. Then there exists a semialgebraic function $F: \pi(S) \rightarrow S$ that satisfies $\pi(F(x)) = x$ for all $x \in \pi(S)$.
\end{theorem}
\begin{proof}
See \cite[Lemma 1.5]{son2016genericity}.
\end{proof}

\section{Semidefinite programming lifts}
\label{sec:sdplifts}

We are now ready to start the proof of Theorem \ref{thm:mainhermitian} (and thus of Theorem \ref{thm:main} too). The first thing we need is a necessary condition for the existence of a semidefinite representation for a given convex set $C$. The condition we state in Theorem \ref{thm:lifts} below is very similar to \cite[Theorem 1]{gouveia2011lifts}\footnote{The given condition is also sufficient, but we will only need necessity here}.
Recall that, for $A \subset \NN^n \times \NN^n$ we denote by $\m_A(z)$ the monomial map:
\[ \m_A(z) = \left[ z^u \bar{z}^v \right]_{(u,v) \in A}. \]
We also denote by $\cl S$ the (topological) closure of a set $S$. 
\begin{theorem}
\label{thm:lifts}
Let $A \subset \NN^n \times \NN^n$. Assume that
\[
\C_A = \cl \conv \left\{ \m_A(z) : z \in \CC^n \right\}
\]
has a semidefinite representation of size $k$. Then there exists $2k^2+1$ semialgebraic functions $f_j:\CC^n \rightarrow \RR$ ($j=1,\ldots,2k^2+1$) such that any nonnegative Hermitian polynomial $p$ supported on $A \cup \{(\mathbf{0},\mathbf{0})\}$ is a sum-of-squares from $V = \linspan_{\RR}(f_1,\ldots,f_{2k^2+1})$, i.e., $p = \sum_{j} h_j^2$ for some $h_j \in V$. Furthermore, the magnitude of the coefficients expressing the $h_j$ in terms of the basis $(f_1,\ldots,f_{2k^2+1})$ are all bounded by $\phi(\|p\|)$ where $\|p\|$ is the largest magnitude of the coefficients of $p$, and $\phi$ is some polynomial that only depends on the semidefinite representation of $\C_A$.
\end{theorem}
The main difference between the statement above and the one in \cite{gouveia2011lifts} (see also \cite[Theorem 5, Chapter 2]{fawziphdthesis}) is that here the functions $f_1,\ldots,f_{2k^2+1}$ are \emph{semialgebraic}. (We say that a function $f:\CC^n\rightarrow \RR$ is semialgebraic if the function $F:\RR^n \times \RR^n \rightarrow \RR$ defined by $F(a,b) = f(a+ib)$ is semialgebraic.) This observation will be crucial to us. We note that a statement similar to the theorem above appears as Theorem 3.4 in \cite{scheidererSDR}. Instead of working with semialgebraic functions, Scheiderer works with polynomial functions on an algebraic variety $X$.
\begin{proof}
Assume that $\C_A = \pi(S)$ where $S$ is a spectrahedron defined using a linear matrix inequality of size $k\times k$:
\[
S = \left\{w \in \RR^N : M(w) := M_0 + M_1 w_1 + \dots + M_N w_N \psd 0 \right\}.
\]
We can assume without loss generality that $S$ has nonempty interior in $\RR^N$. This in turn implies, using standard results about spectrahedra, that there exists $\tilde{w} \in S$ such that $M(\tilde{w})$ is positive definite (possibly after changing $M$), see e.g., \cite[Section 2.4]{ramana1995some}.

For any $z \in \CC^n$ there exists $w(z) \in S$ such that $\pi(w(z)) = \m_A(z)$. Since $M(w(z)) \psd 0$ we can find $F(z) \in \CC^{k\times k}$ such that $M(w(z)) = F(z) F(z)^{\dagger}$. Furthermore, by Theorem \ref{thm:sachoice}, the function $z \in \CC^n \mapsto F(z) \in \CC^{k\times k}$ can be taken to be semialgebraic.

Let $p(z)$ be a Hermitian polynomial supported on $A \cup \{(\mathbf{0},\mathbf{0})\}$, i.e., $p(z) = \langle \pp, \m_A(z) \rangle + c$ for some $c \in \RR$, where $\pp$ denotes the coefficients of the polynomial $p(z)$ in the monomial basis. Since $p \geq 0$ we get $\langle \pp , \m_A(z) \rangle + c \geq 0$ for all $z \in \CC^n$. This implies that $\langle \pp , \sigma \rangle + c \geq 0$ for all $\sigma \in \C_A$. We can lift this linear inequality to an inequality on the spectrahedron $S$, i.e., we have $\langle \pp, \pi(w) \rangle + c \geq 0$ for all $w \in S$, in other words
\[
\langle \pi^*(\pp) , w \rangle + c \geq 0 \quad \forall w \in \RR^N \text{ s.t. } M(w) \psd 0.
\]
By Farkas' lemma/duality for SDPs, this means that there exists $B \psd 0$ and $b \geq 0$ such that 
\begin{equation}
\label{eq:farkasconsequence}
\langle \pi^*(\pp), w \rangle + c = \langle B , M(w) \rangle + b \qquad \forall w \in \RR^N.
\end{equation}
Plugging $w = w(z)$ we get
\[
\langle \pp , \pi(w(z)) \rangle + c = \langle B , M(w(z)) \rangle + b = \langle B , F(z) F(z)^{\dagger} \rangle + b.
\]
Since $\pi(w(z)) = \m_A(z)$ and $\langle \pp , \m_A(z) \rangle = p(z)$ we get finally that $p(z) + c = \langle B , F(z) F(z)^{\dagger}\rangle + b$ for all $z \in \CC^n$. Factorizing $B = DD^{\dagger}$ we get 
\[
p(z) + c = \Tr\left[DD^{\dagger} F(z) F(z)^{\dagger}\right]+b = \sum_{ij=1}^{k} |D^{\dagger} F(z)|^2_{ij}+b = \sum_{ij=1}^{k} \Re[ (D^{\dagger} F(z))_{ij} ]^2 + \Im[ (D^{\dagger} F(z))_{ij} ]^2+b.
\]
If we define the $2k^2+1$ semialgebraic functions to be the constant function and the $z \mapsto \Re[F_{ij}(z)]$ and $z\mapsto \Im[F_{ij}(z)]$, we get the desired claim.

For the last statement of the theorem, we show that the coefficients of $B$ (and thus of $D$) are bounded by a polynomial in the coefficients of $p$. To get this, we can simply plug the value $w=\tilde{w}$ that makes $M(w)$ positive definite in \eqref{eq:farkasconsequence}. If we denote by $\lambda > 0$ the smallest eigenvalue of $M(\tilde{w})$ we get $\langle \pi^*(\pp), \tilde{w} \rangle + c = \langle B , M(\tilde{w}) \rangle + b \geq \lambda \Tr[B] + b \geq \lambda \|B\| + b$, thus $\max(\|B\|,b) \leq (\langle \pi^*(\pp), \tilde{w} \rangle + c)/\min(\lambda,1) \leq O(\max\{\|p\|, |c| \})$.
\end{proof}

\section{Proof of Theorem \ref{thm:mainhermitian}}
\label{sec:proofmainhermitian}

We are now ready to prove our main theorem, Theorem \ref{thm:mainhermitian}. We first recall a piece of notation that we will use throughout the proof: for any function $f:\CC^n \rightarrow \RR$ we associate the function of real variables $F:\RR^n \times \RR^n \rightarrow \RR$, denoted by a capital letter, defined by $F(a,b) = f(a+ib)$. We will also sometimes think of a vector $z \in \CC^n$ as $z \in \RR^{2n}$ and write for instance $F(z)$.

Assume that $\C_A$ has an SDP representation. Then, from Theorem \ref{thm:lifts} there exist semialgebraic functions $f_1,\ldots,f_r:\CC^n \rightarrow \RR$ such that the following is true:
\begin{equation}
\label{eq:propsos}
\tag{$\ast$}
\begin{array}{c}
\text{Any nonnegative Hermitian polynomial supported on $A \cup \{(\mathbf{0},\mathbf{0})\}$}\\
\text{is a sum-of-squares of functions from $\linspan_{\RR}(f_1,\ldots,f_r)$.}
\end{array}
\end{equation}

We will now prove that the functions $f_i$ can be taken to be smooth at the origin. (By this we mean that the associated functions $F_i : \RR^n \times \RR^n \rightarrow \RR$ are smooth at $(0,0)$.) This will follow from our assumption that $A$ is downward closed. Since the $F_i$ are semialgebraic we know from Theorem \ref{thm:sasmooth} that each $F_i$ is smooth almost everywhere. Thus we can find a common point $z_0=(a_0,b_0) \in \RR^n\times \RR^n$ such that all functions $F_i$ are smooth at $z_0$. Now let $\tilde{f_i}(z) = f_i(z-z_0)$ for all $i \in \{1,\ldots,r\}$. We claim that these semialgebraic functions still satisfy the property \eqref{eq:propsos}. Indeed if $q$ is a nonnegative Hermitian polynomial supported on $A \cup \{(\mathbf{0},\mathbf{0})\}$ then $q(z+z_0)$ is nonnegative and is also supported on $A$, since $A$ is downward closed. It follows that $q(z+z_0)$ is a sum-of-squares from $\linspan_{\RR}(f_1,\ldots,f_r)$. But this implies that $q(z)$ is a sum-of-squares from $\linspan_{\RR}(\tilde{f}_1,\ldots,\tilde{f}_r)$.

In the rest of the proof we will thus assume that the $F_1,\ldots,F_r$ are smooth at the origin. If $p$ is homogeneous we are almost done by the following simple observation: if $P$ is a real homogeneous polynomial and $P = \sum_{j} H_j^2$ for some functions $H_j:\RR^{m} \rightarrow \RR$ that are smooth at the origin, then $P$ is a sum-of-squares of \emph{polynomials}. This can be proved by a simple Taylor expansion; for example by applying the following proposition to the identity $t^{2k} P(z) = P(tz) = \sum_{j} H_j(tz)^2$ and observing that $\left.\frac{d^k}{dt^k} H_j(tz)\right|_{t=0}$ is a degree $k$ polynomial in $z \in \RR^{2n}$ (it is the $k$'th term in the Taylor expansion of $H_j$ at $z$).
\begin{proposition}
\label{prop:taylorexpansion}
Assume that $g_j:[0,\eta) \rightarrow \RR$ are smooth functions\footnote{Smoothness at $0$ is smoothness on the right} and that there exists $a \in \RR$ such that $at^{2k} = \sum_{j} g_j(t)^2$ for all $t \in [0,\eta)$. Then $a = \sum_{j} \left(\frac{g_j^{(k)}(0)}{k!}\right)^2$.
\end{proposition}
\begin{proof}
If we Taylor expand the $g_j$ at $0$ we get $at^{2k} = \sum_{j} (g_j(0) + t g_j'(0) + \dots + t^{k} g_j^{(k)}(0) / k! + o(t^k))^2$. By equating powers of $t$ we get that $g_j(0) = \dots = g_j^{(k-1)}(0) = 0$ and that $a = \sum_{j} (g_j^{(k)}(0) / k!)^2$ as desired.
\end{proof}

The case where $p$ is not necessarily homogeneous requires an additional argument. The following argument is inspired from \cite[Proposition 4.18]{scheidererSDR}. Let $2d = \deg p$, and for any $t \in \RR$ consider the Hermitian polynomial $p_{t}(z) = t^{2d} p(z/t)$. This Hermitian polynomial is nonnegative and is also supported on $A$. Thus we know from property \eqref{eq:propsos} that there exist real coefficients $a_j(t) \in \RR^{r}$ s.t. 
\begin{equation}
 \label{eq:hajf}
 P_{t}(z) = \sum_{j} \left ( a_j(t)^{\T} F(z) \right)^2 \qquad \forall z \in \RR^{2n}
\end{equation}
where we let $F(z) = (F_1(z),\ldots,F_{r}(z))$. The functions $a_j(t)$ are defined by a semialgebraic relation and so can be taken to be semialgebraic. As such the $a_j$ must be continuous on some $(0,\eta)$ for $\eta > 0$. From the last part of the statement of Theorem \ref{thm:lifts} we also know that the $a_j$ must be bounded on $(0,\eta)$. Thus, by Theorem \ref{thm:puiseux} we know that the $a_j$ can be extended by continuity to $[0,\eta)$ and that for large enough $m$, $a_j(t^m)$ is smooth on $[0,\eta')$ for some $0 < \eta' < \eta$. From \eqref{eq:hajf} we get:
\[
P_{t^m}(t^{m} z) = \sum_{j} (a_j(t^m)^{\T} F(t^{m} z))^2.
\]
But note that $P_{t^m}(t^{m} z) = t^{2dm} P(z)$. Thus
\[
t^{2dm} P(z) = \sum_{j} (a_j(t^m)^{\T} F(t^m z))^2.
\]
If we let $g_j(t) = a_j(t^m)^{\T} F(t^m z)$ we know that the $g_j$ are smooth on $[0,\eta')$, since the $F$ are smooth at the origin. We can apply the observation of Proposition \ref{prop:taylorexpansion} to get that
\[
P(z) = \sum_{j} \left(\frac{g_j^{(dm)}(0)}{(dm)!}\right)^2.
\]
But, from the definition of $g_j$, $g_j^{(dm)}(0)$ is a polynomial (of degree $d$) in $z$. This contradicts the assumption that $p(z)$ is not a sum of squares of polynomials.

\section{Proof of Theorem \ref{thm:main}}
\label{sec:proofmain}

\newcommand{\p}{p}

\textbf{The case $\Sep(3,3)$:} We prove that $\Sepcone(3,3)$ has no semidefinite representation. Define the Choi polynomial \cite{choi1975positive} by
\[
\begin{aligned}
\p(x,y) &= |x_1|^2 |y_1|^2 + |x_2|^2 |y_2|^2 + |x_3|^2 |y_3|^2 \\
& \quad - 2 ( \Re[x_1 \bar{x_2} y_1 \bar{y_2}] + \Re[x_2 \bar{x_3} y_2 \bar{y_3}] + \Re[x_1 \bar{x_3} y_1 \bar{y_3}] )\\
& \quad + 2 ( |x_1|^2 |y_2|^2 + |x_2|^2 |y_3|^2 + |x_3|^2 |y_1|^2).
\end{aligned}
\]
It was shown in \cite{choi1975positive} (see also \cite[Appendix B]{choi1980some}) that $\p(x,y) \geq 0$ for all $(x,y) \in \CC^3 \times \CC^3$, and yet $\p(x,y)$ is not a sum of squares. In fact the \emph{real} polynomial $\p(x,y)$ when $(x,y) \in \RR^3 \times \RR^3$ is not a sum of squares. It follows, by a simple homogenization argument, that the Hermitian polynomial $\hat{p}(x_1,x_2,y_1,y_2) = p(x_1,x_2,1,y_1,y_2,1)$ is not a sum of squares. Note that the support of $\hat{p}$ satisfies
\[
\supp \hat{p} \subset \hat{A} = \left\{ (\alpha,\beta,\gamma,\delta) \in (\NN^{2} \times \NN^{2}) \times (\NN^{2} \times \NN^{2}) : |\alpha|\leq 1, |\beta|\leq 1, |\gamma|\leq 1, |\delta|\leq 1\right\}.
\]
Since $\hat{A}$ is downward closed it follows from Theorem \ref{thm:mainhermitian} that
\begin{equation}
\label{eq:inhomsep}
\begin{aligned}
\C_{\hat{A}} &= \cl \conv \left\{ \m_{\hat{A}}(x,y) : (x,y) \in \CC^{2} \times \CC^{2} \right\}\\
&=\cl\conv\left\{ \left[x^{\alpha} y^{\beta} \bar{x}^{\gamma} \bar{y}^{\delta}\right]_{|\alpha|\leq 1, |\beta|\leq 1, |\gamma|\leq 1, |\delta|\leq 1} : (x,y) \in \CC^{2} \times \CC^{2} \right\}
\end{aligned}
\end{equation}
does not have a semidefinite representation. To conclude that $\Sepcone(3,3)$ has no semidefinite representation, it remains to note that $\C_{\hat{A}}$ is a hyperplane section of $\Sepcone(3,3)$. Indeed, first recall that $\Sepcone(3,3)$ can be written as
\[
\Sepcone(3,3) = \cl\conv\left\{ \m_{A}(x,y) : (x,y) \in \CC^3 \times \CC^3 \right\}
\]
where
\[
A = \left\{ (\alpha,\beta,\gamma,\delta) \in (\NN^{3} \times \NN^{3}) \times (\NN^{3} \times \NN^{3}) : |\alpha|= 1, |\beta|= 1, |\gamma|= 1, |\delta|= 1\right\}.
\]
It is easy to see that there is a one-to-one correspondence between $\hat{A}$ and $A$. In terms of the monomial map $\m$ this simply means that $\m_{A}$ is the homogenization of $\m_{\hat{A}}$. For example under an appropriate ordering of the monomials we have $\m_{A}(x_1,x_2,x_3,y_1,y_2,y_3) = |x_3|^2 |y_3|^2 \m_{\hat{A}}\left(\frac{x}{x_{3}}, \frac{y}{y_3}\right)$ where we let $x=(x_1,x_2)$ and $y=(y_1,y_2)$. It thus follows that $\Sepcone(3,3)$ can be written as
\[
\Sepcone(3,3) = \cl\conv \left\{ |x_3|^2 |y_3|^2 \m_{\hat{A}}\left(\frac{x}{x_{3}}, \frac{y}{y_3}\right) : (x,y) \in \CC^3 \times \CC^3 \right\}.
\]
It can then be readily verified that $\C_{\hat{A}}$ is a hyperplane section of $\Sepcone(3,3)$ where the appropriate coordinate (corresponding to the monomial $|x_3|^2 |y_3|^2$) is set to 1.

\bigskip

\textbf{The case $\Sep(4,2)$:} Following the same approach as above, we need to exhibit a Hermitian polynomial $\hat{p}(x_1,x_2,x_3,y)$ supported on
\begin{equation}
\label{eq:Ahat31}
\hat{A} = \left\{ (\alpha,\beta,\gamma,\delta) \in (\NN^3 \times \NN) \times (\NN^3 \times \NN) : |\alpha|\leq 1, |\beta|\leq 1, |\gamma|\leq 1, |\delta|\leq 1\right\}
\end{equation}
that is nonnegative but not a sum-of-squares. Since $\Sep(4,2) \neq \PPT(4,2)$ we know that there exists a Hermitian \emph{homogeneous} polynomial $\p(x,y)$ on $(x,y) \in \CC^4 \times \CC^2$ of the form
\[
\p(x,y) = \sum_{ijkl} p_{ijkl} x_i \bar{x_j} y_k \bar{y_l} \quad (x \in \CC^4, y \in \CC^2)
\]
that is nonnegative but not a sum of squares. Note that such a $p$ satisfies $p(\lambda x, \mu y) = |\lambda|^2 |\mu|^2 p(x,y)$ for any $(\lambda,\mu) \in \CC^2$. To get the desired polynomial $\hat{p}$ it would suffice to dehomogenize the polynomial $\p$ by setting one of the $x$ variables to 1,  and one of the $y$ variables to 1. It turns out, however, that one cannot guarantee in general that this dehomogenized polynomial is not a sum of squares. (We give an explicit example in Appendix \ref{sec:hakye}.) The reason we could dehomogenize the Choi polynomial in the $(3,3)$ case was that the Choi polynomial is not a sum of squares  when the variables are \emph{real}. One cannot expect this to be true for our polynomial in (4,2) variables as it is known that any biquadratic \emph{real} polynomial in $(n,2)$ variables \emph{is} a sum of squares \cite{calderon}. Nevertheless we show in Appendix \ref{sec:hakye} that by choosing an appropriate polynomial $\p$, and an appropriate dehomogenization we can get a polynomial $\hat{p}(x,y)$ supported on $\hat{A}$ of Equation \eqref{eq:Ahat31} that is not a sum-of-squares. This implies that $\C_{\hat{A}}$ is not semidefinite representable. Using a similar argument as for the $(3,3)$ case we get that $\Sepcone(4,2)$ is not semidefinite representable.

\appendix

\section{The case $\Sep(4,2)$}
\label{sec:hakye}

In this section we exhibit a nonnegative Hermitian polynomial $\hat{p}(x_1,x_2,x_3,y)$ supported on $\hat{A}$, defined in Equation \eqref{eq:Ahat31}, that is not a sum of squares.

Consider the following map $\Phi:\M_2 \rightarrow \M_4$ studied in \cite{hakye}:
\[
\Phi\left(\begin{bmatrix} x & y\\ z & w\end{bmatrix}\right)
=
\left[
\begin{array}{cccc}
 3 w+4 x-2 y-2 z & 2 z-2 x & 0 & 0 \\
 2 y-2 x & 2 x & z & 0 \\
 0 & y & 2 w & -w-2 z \\
 0 & 0 & -w-2 y & 2 w+4 x \\
\end{array}
\right].
\]
It is shown in \cite{hakye} that the map $\Phi$ is positive but not decomposable. 
We associate to $\Phi$ the Hermitian polynomial
\[
\W(x,y) = x^{\dagger} \Phi(yy^{\dagger}) x \qquad x \in \CC^4, y \in \CC^2.
\]
Then we know from Proposition \ref{prop:decompsos} that $\W$ is positive but not a sum of squares. The purpose of this section is to prove the following:
\begin{proposition}
\label{prop:Wdehom}
The (nonhomogeneous) Hermitian polynomial $\W(1,x_2,x_3,x_4,1,y_2)$ is not a sum of squares.
\end{proposition}
The proof of this proposition involves some computations, and we use the specific properties of $\W$ (its zeros) which have been studied in \cite{hakye}. We note that there are other dehomogenizations of $\W$ that \emph{are} sums of squares. For example, we show later that $\W(x_1,x_2,-6,x_4,1,y_2)$ \emph{is} a sum of squares.
\begin{proof}[Proof of Proposition \ref{prop:Wdehom}]
To lighten the notation we let $y_2 = \alpha$. In \cite{hakye}, the zeros of the polynomial $\W$ were identified. Namely it was shown that
\begin{equation}
\label{eq:zeroW}
\W ( x_1(\alpha) , x_2(\alpha) , x_3(\alpha) , x_4(\alpha) , 1 , \alpha ) = 0 \quad \forall \alpha \in \CC
\end{equation}
where
\[
x(\alpha) := \Bigl(2\alpha(1-\alpha) , \;\; \alpha \left[ 4 - 2 (\alpha+\bar{\alpha}) + 3 |\alpha|^2 \right] , \;\; -4 - 2|\alpha|^2 , \;\; -\bar{\alpha}(2+\alpha)\Bigr) \in \CC^4.
\]
Let $p(x_2,x_3,x_4,\alpha) = \W(1,x_2,x_3,x_4,1,\alpha)$. The explicit formula of $p$ is
\[
\begin{aligned}
p(x_2,x_3,x_4,\alpha) &= 2 |\alpha|^2 |x_3|^2 + 2 |\alpha|^2 |x_4|^2 - x_3 \bar{x_4} |\alpha|^2  - \bar{x_3}  x_4 |\alpha|^2 \\
  &\quad+ \bar{x_2} x_3 \alpha + x_2 \bar{x_3}\bar{\alpha} - 2 \bar{x_3} x_4 \alpha - 2 x_3 \bar{x_4}\bar{\alpha} \\
  &\quad+ 3 |\alpha|^2  + 2 |x_2|^2 + 4 |x_4|^2 + 2 x_2\alpha  + 2 \bar{x_2}\bar{\alpha}\\
  &\quad - 2 x_2 - 2\bar{x_2} - 2 \alpha - 2\bar{\alpha} + 4.
\end{aligned}
\]
Assume that $p = \sum_{i} g_i^2$ where $g_i$ are Hermitian polynomials in $x_2,x_3,x_4,\alpha$. Since $p$ has no terms $|x_3|^2$ we see that the $g_i$ cannot contain monomials $x_3$ or $\bar{x_3}$. Similarly $p$ does not have a term $|\alpha|^2 |x_2|^2$ and so $p$ cannot contain monomials $\alpha x_2$, $\bar{\alpha} x_2$, $\alpha \bar{x_2}$ or $\bar{\alpha} \bar{x_2}$. It follows that each $g_i$ must be a linear combination of the monomials
\[
1, \;\; \alpha, \;\; x_2, \;\; x_4, \;\; \alpha x_3, \;\; \alpha x_4, \;\; \bar{\alpha} x_3, \;\; \bar{\alpha} x_4
\]
and their conjugates. In other words, each $g_i$ is of the form:
\begin{equation}
\label{eq:g_i_expansion}
\begin{aligned}
g_i(x_2,x_3,x_4,\alpha) &= a_i + b_i \alpha + c_i x_2 + d_i x_4 + e_i \alpha x_3 + f_i \alpha x_4 + g_i \bar{\alpha} x_3 + h_i \bar{\alpha} x_4\\
&\quad + \bar{b_i} \bar{\alpha} + \bar{c_i} \bar{x_2} + \bar{d_i} \bar{x_4} + \bar{e_i} \bar{\alpha} \bar{x_3} + \bar{f_i} \bar{\alpha} \bar{x_4} + \bar{g_i} \alpha \bar{x_3} + \bar{h_i} \alpha \bar{x_4}
\end{aligned}
\end{equation}
where $a_i,b_i,\ldots,h_i \in \CC$. We will now use the information about the zeros of $\W$ (and thus of $p$) to deduce relations about these coefficients and reach a contradiction.

Since $\W$ is bihomogeneous in the first set of variables, we have (dividing by $|x_1(\alpha)|^2$) from \eqref{eq:zeroW} that
\[
p \left( \;\; \frac{x_2(\alpha)}{x_1(\alpha)}, \;\; \frac{x_3(\alpha)}{x_1(\alpha)}, \;\; \frac{x_4(\alpha)}{x_1(\alpha)} , \;\; \alpha \right) = 0, \quad \forall \alpha \in \CC \setminus \{0,1\}.
\]
Since $p = \sum_{i} g_i^2$ we get that for all $i$,
\begin{equation}
\label{eq:eqzerog_i}
g_i \left( \;\; \frac{x_2(\alpha)}{x_1(\alpha)}, \;\; \frac{x_3(\alpha)}{x_1(\alpha)}, \;\; \frac{x_4(\alpha)}{x_1(\alpha)} , \;\; \alpha \right) = 0, \quad \forall \alpha \in \CC \setminus \{0,1\}.
\end{equation}
We can clear denominators in \eqref{eq:eqzerog_i} by multiplying the expression by $|x_1(\alpha)|^2$.  As a result we get that
\begin{equation}
\label{eq:eqzerog_i2}
|x_1(\alpha)|^2 g_i \left( \;\; \frac{x_2(\alpha)}{x_1(\alpha)}, \;\; \frac{x_3(\alpha)}{x_1(\alpha)}, \;\; \frac{x_4(\alpha)}{x_1(\alpha)} , \;\; \alpha \right) = 0, \quad \forall \alpha \in \CC.
\end{equation}
The left-hand side of \eqref{eq:eqzerog_i2} is a Hermitian polynomial in $\alpha$ that is identically zero. Hence all its coefficients must be equal to 0. This allows us to derive conditions on the coefficients of $g_i$ in \eqref{eq:g_i_expansion}. More precisely:
\begin{itemize}
\item The coefficient of $\alpha^4$ is $4\bar{h_i}$. Setting $4\bar{h_i}$ to zero yields $h_i = 0$.
\item The coefficient of $\alpha^4 \bar{\alpha}$ is $4 \bar{g_i} + 2 \bar{h_i}$. Setting to zero we get $g_i = 0$.
\item The coefficient of $\alpha^2$ is $-4 \bar{d_i} -8\bar{g_i}$. Setting to zero we get $d_i = 0$.
\end{itemize}
This gives a contradiction: indeed the coefficient of $|x_4|^2$ in $p$ is $4 > 0$ and yet $\sum_{i} |d_i|^2 = 0$.
\end{proof}

We conclude this appendix by proving, as promised, that there is another dehomogenization of $\W$ that is a sum of squares. Let:
\[
q(x_1,x_2,x_4,\alpha) = \W(x_1,x_2,-6,x_4,1,\alpha).
\]

Let
\[
A = 
\left(
\begin{array}{cccccc}
 36 & 0 & -3 & 0 & 0 & 0 \\
 0 & 2 & -1 & 0 & -1 & 0 \\
 -3 & -1 & 1 & 0 & 1 & 0 \\
 0 & 0 & 0 & 2 & 0 & 0 \\
 0 & -1 & 1 & 0 & \frac{3}{2} & 0 \\
 0 & 0 & 0 & 0 & 0 & 1 \\
\end{array}
\right),
\qquad
B = \left(
\begin{array}{cccccc}
 0 & 0 & 0 & 6 & 0 & 3 \\
 0 & 0 & 0 & 0 & 0 & -1 \\
 0 & 0 & 0 & 0 & 0 & 0 \\
 6 & 0 & 0 & 0 & 1 & 0 \\
 0 & 0 & 0 & 1 & 0 & 0 \\
 3 & -1 & 0 & 0 & 0 & 0 \\
\end{array}
\right).
\]

One can verify that $A-B \psd 0$ and $A+B \psd 0$, i.e., that $\begin{sm} A & B\\ B & A\end{sm} \psd 0$. Let $m(x_1,x_2,x_4,\alpha) = (\alpha,x_1,x_2,x_4,\bar{\alpha}x_1,\bar{\alpha}x_4)$. Then one can check that we have the following sum of squares decomposition of $q$:
\[
q(x_1,x_2,x_4,\alpha) = \begin{bmatrix} m(x,\alpha)\\ \bar{m}(x,\alpha) \end{bmatrix}^{\dagger} \begin{bmatrix} A & B\\ B & A \end{bmatrix} \begin{bmatrix} m(x,\alpha)\\ \bar{m}(x,\alpha) \end{bmatrix}.
\]

\paragraph{Acknowledgements} I would like to thank Omar Fawzi for his encouragements and for helpful comments on the paper. I would also like to thank Claus Scheiderer for useful discussions and exchanges related to the material of this paper, and James Saunderson for comments that helped improved the exposition.

\bibliography{../../bib/nonnegative_rank}
\bibliographystyle{alpha}

\end{document}